\documentclass[conference]{IEEEtran}

\newcounter{remark}[section]

\usepackage{amsmath}
\usepackage{amssymb}
\usepackage{bm}
\usepackage{mathrsfs}
\usepackage[dvips]{graphicx}

\usepackage{color}


\newtheorem{bigthm}{Theorem}
\newtheorem{thm}{Theorem}
\newtheorem{cor}[thm]{Corollary}
\newtheorem{lemma}[thm]{Lemma}
\newtheorem{prop}[thm]{Proposition}

\DeclareMathAlphabet{\mathsfsl}{OT1}{cmss}{m}{sl}

\newcommand{\term}{\emph}

\newcommand{\cnst}[1]{\mathrm{#1}}

    {\makebox{\phantom{}} \\ %
        \textsc{Input:}\begin{itemize}}%
    {\end{itemize}}

    {\textsc{Output:}\begin{itemize}}%
    {\end{itemize}}

    {\textsc{Procedure:}\begin{enumerate}}%
    {\end{enumerate}}

\renewcommand{\phi}{\varphi}

\newcommand{\Id}{\mathbf{I}}

\newcommand{\coll}[1]{\mathscr{#1}}

\newcommand{\Cspace}[1]{\mathbb{C}^{#1}}

\newcommand{\abs}[1]{\left\vert {#1} \right\vert}
\newcommand{\abssq}[1]{{\abs{#1}}^2}

\newcommand{\Prob}[1]{\mathbb{P}\left\{ {#1} \right\}}

\newcommand{\vct}[1]{\bm{#1}}
\newcommand{\mtx}[1]{\bm{#1}}

\newcommand{\adj}{*}
\newcommand{\psinv}{\dagger}

\newcommand{\range}{\operatorname{range}}

\newcommand{\rank}{\operatorname{rank}}

\newcommand{\trace}{\operatorname{trace}}

\newcommand{\supp}[1]{\operatorname{supp}(#1)}

\newcommand{\ip}[2]{\left\langle {#1},\ {#2} \right\rangle}
\newcommand{\absip}[2]{\abs{\ip{#1}{#2}}}

\newcommand{\norm}[1]{\left\Vert {#1} \right\Vert}
\newcommand{\normsq}[1]{\norm{#1}^2}

\newcommand{\smnorm}[2]{{\bigl\Vert {#2} \bigr\Vert}_{#1}}

\newcommand{\enorm}[1]{\norm{#1}_2}
\newcommand{\enormsq}[1]{\enorm{#1}^2}

\newcommand{\fnorm}[1]{\norm{#1}_{\mathrm{F}}}
\newcommand{\fnormsq}[1]{\fnorm{#1}^2}

\newcommand{\pnorm}[2]{\norm{#2}_{#1}}

\newcommand{\subjto}{\quad\text{subject to}\quad}

\newcommand{\bigO}{\mathrm{O}}

\newcommand{\atom}{\vct{\phi}}
\newcommand{\Fee}{\mtx{\Phi}}

\begin{document}
\title{The Sparsity Gap: \\ Uncertainty Principles Proportional to Dimension}

\author{\IEEEauthorblockN{Joel A.~Tropp}
\IEEEauthorblockA{Computing and Mathematical Sciences\\
California Institute of Technology\\
Pasadena, CA 91125--5000\\
Email: jtropp@acm.caltech.edu}}

\maketitle

\begin{abstract}
In an incoherent dictionary, most signals that admit a sparse representation admit a unique sparse representation.  In other words, there is no way to express the signal without using strictly more atoms. This work demonstrates that sparse signals typically enjoy a higher privilege: each nonoptimal representation of the signal requires far more atoms than the sparsest representation---unless it contains many of the same atoms as the sparsest representation.  One impact of this finding is to confer a certain degree of legitimacy on the particular atoms that appear in a sparse representation.  This result can also be viewed as an uncertainty principle for random sparse signals over an incoherent dictionary.  
\end{abstract}

\IEEEpeerreviewmaketitle

\section{Introduction}

The purpose of this paper is to develop a new class of uncertainty principles for sparse representation that hold even when the sparsity level approaches the ambient dimension.  We begin with a discussion of the background and related results before moving on to the new contributions.

\subsection{Sparse Representation in Dictionaries}

Let $\Fee$ be an $m \times N$ matrix with normalized columns:
$$
\enorm{ \atom_j } = 1,
\quad j = 1, 2, \dots, N.
$$
We refer to $\Fee$ as a \term{dictionary} and to its columns as \term{atoms}.  Assume the atoms span the ambient space $\Cspace{m}$.

There are two simple geometric quantities associated with a dictionary.  The first is a measure of \term{redundancy}:
$$
\rho = \norm{\Fee}^2,
$$
where $\norm{\cdot}$ denotes the spectral, or $\ell_2 \to \ell_2$ operator norm, of a matrix.  We always have $\rho \geq N/m$.  Equality holds if and only if $\Fee$ is a \term{tight frame}. 
The second quantity is the \term{coherence}:
$$
\mu = \max_{j \neq k} \absip{ \atom_j }{ \atom_k }.
$$
The coherence is small when the angle between each pair of atoms is large.  Strohmer and Heath~\cite{SH03:Grassmannian-Frames} have observed that 
\begin{equation} \label{eqn:grassmann}
\mu \geq \sqrt{\frac{N - m}{m(N-1)}}.
\end{equation}
In the typical case $N \geq 2m$, the inequality~\eqref{eqn:grassmann} indicates that the coherence cannot be very small: $\mu \gtrsim m^{-1/2}$.

Let $S$ be a subset of $\{1,2,\dots, N\}$, and define $\Fee_S$ to be the column submatrix of $\Fee$ whose columns are listed in $S$.  
We say that $S$ is \term{linearly independent} if it lists a linearly independent family of atoms.  Note that $\Fee_S$ is injective if and only if $S$ is linearly independent.  Suppose that a signal $\vct{u} \in \Cspace{m}$ can be written as
$$
\vct{u} = \Fee \vct{x}
\quad\text{where $\supp{\vct{x}} \subset S$}.
$$
We call the vector $\vct{x}$ a \term{representation} of the signal $\vct{u}$, and we say that $\vct{u}$ can be \term{represented with $S$.}  When $S$ is linearly independent, $\vct{x}$ is the unique representation of $\vct{u}$ over $S$.

In a redundant dictionary ($N > m$), each signal has an infinity of representations.  The sparse representation problem asks us to express $\vct{u}$ with the fewest number of atoms:
\begin{equation} \label{eqn:p0}
\min \pnorm{0}{\vct{z}}
\quad\subjto\quad
\vct{u} = \Fee \vct{z},
\end{equation}
where 
$\pnorm{0}{\cdot}$ counts the number of nonzero components in its argument.  If $\vct{x}$ is a minimizer of this mathematical program, the set $S = \supp{\vct{x}}$ must be linearly independent.  Otherwise, we could remove an atom to obtain a sparser representation.  As a result, when studying sparse representation, we focus on linearly independent sets of atoms.


\subsection{Uniqueness of Sparse Representations} \label{sec:uniqueness}

One might wonder when the problem~\eqref{eqn:p0} has a unique solution.  The sparse approximation literature took up this inquiry about ten years ago, although one can trace some of the ideas to the late 1980s~\cite{DS89:Uncertainty-Principles}. 
The early research led to the following result for deterministic signals.

\begin{prop} \label{prop:l0-uniqueness}
Assume that
\begin{equation} \label{eqn:l0-uniqueness}
\abs{S} < \frac{1}{2}( \mu^{-1} + 1 ).
\end{equation}
If a signal $\vct{u} = \Fee \vct{x}$ with $\supp{\vct{x}} \subset S$, then $\vct{x}$ is the unique minimizer of~\eqref{eqn:p0}.
\end{prop}

Donoho and Huo established this result for dictionaries consisting of two orthonormal bases~\cite{DH01:Uncertainty-Principles}; Gribonval and Nielsen proved that it holds for \emph{every} dictionary~\cite{GN03:Sparse-Representations}.  Subsequently, Donoho and Elad showed that Proposition~\ref{prop:l0-uniqueness} follows from a more general result, phrased in terms of the \term{Kruskal rank}, or \term{spark}, of the dictionary~\cite{DE03:Optimally-Sparse}.  Another line of work~\cite{EB02:Generalized-Uncertainty,GN03:Sparse-Representations} sharpened the condition~\eqref{eqn:l0-uniqueness} for dictionaries consisting of multiple orthonormal bases.  See~\cite{Tro08:Linear-Independence} for a detailed discussion of the spikes and sines dictionary.

The requirement~\eqref{eqn:l0-uniqueness} is very stringent: it typically demands that the sparsity level $\abs{ S } \lesssim \sqrt{m}$.  In spite of this apparent shortcoming, the condition~\eqref{eqn:l0-uniqueness} cannot be improved in general.  For example, when $m$ is a perfect square, the Dirac comb can be represented perfectly using $\sqrt{m}$ spikes or $\sqrt{m}$ sines~\cite{DS89:Uncertainty-Principles}.  To move past the square-root threshold, we must place additional restrictions on the sparse signals we are willing to consider.

A natural approach is to introduce some randomness.  Let $S$ be linearly independent, and let $\vct{x} \in \Cspace{S}$ be a random vector whose distribution is absolutely continuous with respect to the Lebesgue measure on $\Cspace{S}$.  We say that a random signal of the form $\vct{u} = \Fee_S \vct{x}$ is \term{generic}, and we refer to the (unique) representation of $\vct{u}$ over $S$ as the \emph{natural representation}.  


To obtain interesting uniqueness results for generic signals, we impose some additional hypotheses. 
We say that the dictionary $\Fee$ is a \term{weakly incoherent tight frame} if
\begin{equation}\label{eqn:weak-incoherence-intro}
\normsq{\Fee} = \frac{N}{m}
\quad\text{and}\quad
\mu \leq \frac{\cnst{c}}{\log N},
\end{equation}
where $\cnst{c}$ is an absolute constant.  Assume the sparsity level
\begin{equation} \label{eqn:s-bound-intro}
s \leq \frac{\cnst{c} m}{\log N}.
\end{equation}
In this setting, we have the following result.

\begin{prop} \label{prop:weak-unique}
Assume the dictionary $\Fee$ satisfies~\eqref{eqn:weak-incoherence-intro} and the sparsity $s$ satisfies~\eqref{eqn:s-bound-intro}.  Draw a uniformly random set $S$ of $s$ atoms from the dictionary.  Except with probability $\bigO(N^{-1})$, the following statement holds.

Let $\vct{u} = \Fee_S \vct{x}$ be a generic signal.  With probability one, the natural representation of $\vct{u}$ is the unique minimizer of~\eqref{eqn:p0}.
\end{prop}

Roughly speaking, Proposition~\ref{prop:weak-unique} states that a generic sparse signal over a random set of atoms is unlikely to have any other representation that is equally sparse---even when the sparsity level is nearly proportional to the ambient dimension.  

Cand\`es and Romberg established the first theorem of this type in the specific case of the spikes and sines dictionary~\cite{CR06:Quantitative-Robust}.  Using different methods, the present author showed that analogous results hold for any strongly incoherent dictionary~\cite[Sec.~7]{Tro08:Conditioning-Random}.  The extension to weakly incoherent dictionaries requires additional ideas from~\cite[Sec.~5]{Tro08:Norms-Random}.  

\subsection{Uncertainty Principles}

Historically, the sparse approximation community has viewed uniqueness through the lens of \emph{uncertainty principles}.  Suppose that a signal has two (different) representations:
$$
\vct{u} = \Fee_S \vct{x} = \Fee_T \vct{y}.
$$
The Donoho--Elad dictionary uncertainty principle~\cite[Thms.~3 and~5]{DE03:Optimally-Sparse} states%
\footnote{Donoho and Elad express their uncertainty principle~\cite[Thm.~3]{DE03:Optimally-Sparse} in terms of the Kruskal rank of the dictionary, which is notoriously difficult to estimate.  The result quoted here provides the best general bound.}
that
\begin{equation} \label{eqn:donoho-elad}
\abs{S} + \abs{T} > \mu^{-1}.
\end{equation}
In particular, if a signal $\vct{u}$ can be represented with a set $S$ that satisfies~\eqref{eqn:l0-uniqueness}, then every alternative representation requires strictly more atoms.

Since the coherence usually satisfies $\mu \gtrsim m^{-1/2}$, the dictionary uncertainty principle only operates in the regime of very sparse representations: $\abs{S} \lesssim \sqrt{m}$.  Except for very structured (or very random) dictionaries, it does not seem possible to obtain dictionary uncertainty principles for arbitrary signals that hold at sparsity levels near the ambient dimension.


\subsection{The Sparsity Gap}

This paper describes uncertainty principles for \emph{generic signals}.  It is easy to appreciate why generic signals might behave better than adversarially chosen signals.  If there are small sets $S$ and $T$ of atoms for which $\range(\Fee_S)$ and $\range(\Fee_T)$ intersect, there exists a signal that has sparse representations over both $S$ and $T$.  (Witness the Dirac comb!)  On the other hand, it is hard for a generic signal to have two sparse representations because $\range(\Fee_T)$ rarely \emph{contains} $\range(\Fee_S)$!  This fact offers a plausible route to reach uncertainty principles at sparsity levels far greater than $\sqrt{m}$.

Our first result extends the dictionary uncertainty principle~\eqref{eqn:donoho-elad} to generic signals.  The proof appears in Section~\ref{sec:strong-incoherence}.

\begin{bigthm}[Sparsity Gap under Strong Incoherence] \label{thm:strong-intro}
Suppose that $S$ is linearly independent, and draw a generic signal $\vct{u}$ in $\range(\Fee_S)$.  Then, almost surely, we cannot represent $\vct{u}$ with a set $T$ disjoint from $S$ \emph{unless}
$$
\abs{S} + \abs{T} > \mu^{-1} \sqrt{ \abs{S} }.
$$
\end{bigthm}

When $\abs{S} = 1$, our result coincides with the dictionary uncertainty principle~\eqref{eqn:donoho-elad}, but it becomes increasingly strict requirements as the sparsity level $\abs{S}$ increases!  Indeed, an equivalent condition is
$$
\abs{T} > \left( \frac{\mu^{-1}}{\sqrt{\abs{S}}} - 1 \right) \cdot \abs{S},
$$
so a generic signal that uses $\abs{S} \ll \mu^{-2}$ atoms cannot be represented with any disjoint set $T$ of atoms unless $\abs{T} \gg \abs{S}$.  
In the extreme case where $\mu^{-2} \sim m$, we obtain an uncertainty principle that operates at sparsity levels proportional to the ambient dimension!

Our second result is an uncertainty principle that parallels Proposition~\ref{prop:weak-unique}, just as Theorem~\ref{thm:strong-intro} parallels Proposition~\ref{prop:l0-uniqueness}.  This proof appears in Section~\ref{sec:weak-incoherence}.

\begin{bigthm}[Sparsity Gap under Weak Incoherence] \label{thm:weak-intro}
Assume that $\Fee$ is a weakly incoherent tight frame that satisfies~\eqref{eqn:weak-incoherence-intro}, and assume further that $N > 2m$.  Suppose that $S$ is a randomly chosen set of $s$ atoms, where $s$ satisfies~\eqref{eqn:s-bound-intro}.  Except with probability $\bigO(N^{-1})$, the following holds.

Draw a generic signal $\vct{u}$ in $\range(\Fee_S)$.  Then, almost surely, $\vct{u}$ cannot be represented with a set $T$ disjoint from $S$ \emph{unless}
$$
\abs{T} > \left( 1 + \frac{2}{\rho} \right) \cdot \abs{S}. 
$$
The redundancy $\rho = N/m$, by hypothesis.
\end{bigthm}

In words, Theorem~\ref{thm:weak-intro} considers a generic signal over a random set of atoms.  It is likely that every (disjoint) alternative representation requires a constant factor more atoms than the natural representation, where the extra factor decreases as the dictionary becomes more redundant.  We see that there is typically a \emph{sparsity gap} between the natural representation and the the sparsest representation that uses different atoms. 

This result provides an interesting guarantee for a huge class of dictionaries because of the weak bound for the incoherence.  On the other hand, it holds for a smaller class of signals than Theorem~\ref{thm:strong-intro} because we have randomized the set of atoms in addition to choosing generic coefficients.

\section{Rank and File}

Although a generic signal has many representations aside from the natural one, there is a large class of representations that we can almost surely rule out.  As a first step toward our main results, we develop an algebraic condition that describes which representations can and cannot occur.

To motivate the discussion, let us recall a standard argument for establishing dictionary uncertainty principles.  Suppose that both $S$ and $T$ are linearly independent.  A few moments of thought reveals that the following conditions are equivalent:
\begin{enumerate}
\item	We have $\range(\Fee_S) \cap \range(\Fee_T) = \emptyset$.
\item	The matrix $\Fee_R$ has full rank, where $R = S \cup T$.
\end{enumerate}
For a fixed set $S$, suppose that Condition 2) holds whenever $\abs{R} < r_{\star}$.  We conclude that, if there exists a signal that has representations over both $S$ and $T$, then $\abs{S} + \abs{T} \geq r_{\star}$.  Read the paper~\cite{DE03:Optimally-Sparse} to see this argument in action.

We can extend this methodology by \emph{quantifying} the rank of the matrix $\Fee_R$.  These bounds allow us to count how many extra atoms are needed to represent a generic sparse signal.

\begin{lemma} \label{lem:subspace-cond}
Suppose that both $S$ and $T$ are linearly independent.  The following conditions are equivalent.
\begin{enumerate}
\item	We have $\range(\Fee_S) \cap \range(\Fee_T) \subsetneq \range(\Fee_S)$.
\item	We have $\abs{T} < \rank( \Fee_R )$, where $R = S \cup T$.
\end{enumerate}
\end{lemma}

\begin{IEEEproof}
Define the subspaces $\coll{S} = \range(\Fee_S)$ and $\coll{T} = \range(\Fee_T)$, which implies that $\coll{S} + \coll{T} = \range(\Fee_R)$.  Note that $\coll{S} \cap \coll{T}$ is a proper subspace of $\coll{S}$ if and only if
\begin{equation} \label{eqn:subspace-*}
\dim( \coll{S} \cap \coll{T} ) < \dim( \coll{S} ) 
\end{equation}
The algebra of subspaces yields
$$
\dim( \coll{S} \cap \coll{T} ) = \dim(\coll{S}) + \dim(\coll{T}) - \dim( \coll{S} + \coll{T} ).
$$
Therefore, the condition~\eqref{eqn:subspace-*} is equivalent with
$$
\dim(\coll{T}) < \dim(\coll{S} + \coll{T}).
$$
Since $T$ is linearly independent, $\dim(\coll{T}) = \abs{T}$.  Meanwhile,
$$
\dim(\coll{S} + \coll{T}) = \dim( \range(\Fee_R) ) \\
	= \rank(\Fee_R).
$$
This is the required conclusion. 
\end{IEEEproof}

Let us translate the previous result from the language of subspaces to the language of probability.

\begin{cor} \label{cor:generic-signal}
Suppose that both $S$ and $T$ are linearly independent.  The following conditions are equivalent.
\begin{enumerate}
\item	A generic signal $\vct{u} = \Fee_S \vct{x}$ almost surely has no representation of the form $\vct{u} = \Fee_T \vct{y}$.

\item	We have $\abs{T} < \rank( \Fee_R )$, where $R = S \cup T$.
\end{enumerate}
\end{cor}

\begin{IEEEproof}
Lemma~\ref{lem:subspace-cond} states that Condition 2) is the same as
\begin{equation} \label{eqn:range-*}
\range(\Fee_S) \cap \range(\Fee_T) \subsetneq \range(\Fee_S), 
\end{equation}
so we prove that Condition 1) is the same as \eqref{eqn:range-*}.  To that end, let $\vct{u} = \Fee_S \vct{x}$ be a generic signal, which means that $\vct{x}$ is absolutely continuous with respect to the Lebesgue measure on $\Cspace{S}$.  Let $\nu$ denote the Lebesgue measure on $\range(\Fee_S)$.

First, assume \eqref{eqn:range-*} holds.  A proper subspace has zero Lebesgue measure, so
$$
\nu( \range(\Fee_S) \cap \range(\Fee_T) ) = 0.
$$
The set $S$ is linearly independent, so $\Fee_S$ is injective.  As a result, the distribution of $\vct{u}$ is absolutely continuous with respect to $\nu$.  It follows immediately that
$$
\Prob{ \vct{u} \in \range(\Fee_S) \cap \range(\Fee_T) } = 0.
$$
We conclude that
$$
\Prob{ \vct{u} \in \range(\Fee_T) } = 0
$$
because the signal $\vct{u} \in \range(\Fee_S)$.

Conversely, suppose \eqref{eqn:range-*} is false.  Then $\range(\Fee_S) \subset \range(\Fee_T)$, so the signal $\vct{u}$ can be represented over $T$.
\end{IEEEproof}

It is convenient to remove the assumption of linear independence from the previous result.

\begin{cor} \label{cor:generic-ld}
Suppose that $S$ is linearly independent, and let $T$ be any other set of atoms.  Assume that
$$
\abs{T} < \rank( \Fee_{R} ),
\quad\text{where $R = S \cup T$.}
$$
Draw a generic signal $\vct{u} = \Fee_S \vct{x}$.  Then
$$
\Prob{ \vct{u} \in \range(\Fee_T) } = 0.
$$
\end{cor}

\begin{proof}
When $T$ is linearly independent, the claim follows directly from Corollary~\ref{cor:generic-signal}.  Otherwise, extract a maximal linear independent subset $T'$ from $T$, and write $R' = S \cup T'$.  Apply the result to $T'$ to obtain the statement
$$
\abs{T'} < \rank(\Fee_{R'}) \quad\Longrightarrow\quad
\Prob{ \vct{u} \in \range(\Fee_{T'}) } = 0.
$$
Since $T'$ is maximal, $\rank(\Fee_{R'}) = \rank(\Fee_R)$ and also $\range(\Fee_{T'}) = \range(\Fee_T)$.  To complete the proof, note that the hypothesis $\abs{T} < \rank(\Fee_R)$ implies $\abs{T'} < \rank(\Fee_R)$ because $T' \subset T$.
\end{proof}

\section{Analytic Rank Estimates}

The main challenge is that we only possess analytic/geometric information about the dictionary, encapsulated in the redundancy $\rho$ and the coherence $\mu$.  But the rank is fundamentally an algebraic quantity.  Our approach will be to construct \emph{analytic} estimates for the rank that we can compute from the data at hand.

\subsection{Schatten Norms}

A primary tool is the Schatten class of matrix norms.
Let $\mtx{A}$ be a matrix, and write $\vct{\sigma}(\mtx{A})$ for the vector of singular values of $\mtx{A}$, arranged in weakly decreasing order.  The \term{Schatten $p$-norm} is defined as
$$
\pnorm{S_p}{\mtx{A}} = \pnorm{p}{ \vct{\sigma}(\mtx{A}) },
$$
where $\pnorm{p}{\cdot}$ is the usual $\ell_p$ vector norm.  In particular, $S_2$ is the Frobenius norm, and $S_\infty$ is the spectral norm.  The norm $S_1$ is often called the \term{trace norm} because
$$
\pnorm{S_1}{\mtx{A}} = \trace(\mtx{A})
\quad\text{when $\mtx{A}$ is psd}.
$$
The term \term{psd} abbreviates \term{positive semidefinite}.  For general matrices, the Frobenius norm is the only Schatten-class norm computable directly from the matrix entries:
$$
\fnorm{\mtx{A}} = \left[ \sum\nolimits_{jk} \abssq{a_{jk}} \right]^{1/2}.
$$




\subsection{Rank Bounds via Norm Ratios}

A simple but powerful method for estimating rank is to compare two different Schatten norms of the same matrix.

\begin{lemma}
Suppose that $p < q$.  For each matrix $\mtx{A}$,
$$
\rank(\mtx{A}) \geq \left[ \frac{ \pnorm{S_p}{\mtx{A}} }{ \pnorm{S_q}{\mtx{A}}} \right]^{pq/(q-p)}.
$$
\end{lemma}

\begin{proof}
For each vector $\vct{x} \in \Cspace{r}$, we have the inequality 
$$
\frac{\pnorm{p}{\vct{x}}}{\pnorm{q}{\vct{x}}} \leq r^{1/p - 1/q}.
$$
Indeed, one can use Lagrange multipliers to verify that the left-hand side is maximized when $\vct{x}$ is a constant vector.

Suppose that $\rank(\mtx{A}) = r$.  Then the vector $\vct{\sigma}$ of nonzero singular values of $\mtx{A}$ lies in $\Cspace{r}$.  By definition of the Schatten norms,
$$
\frac{\pnorm{S_p}{\mtx{A}}}{\pnorm{S_q}{\mtx{A}}}
	= \frac{\pnorm{p}{\vct{\sigma}}}{\pnorm{q}{\vct{\sigma}}}
	\leq r^{1/p - 1/q}.
$$
Take the $(1/p - 1/q)$ root and simplify the exponent to reach the conclusion.
\end{proof}

The following simple corollary is fantastically useful.

\begin{cor} \label{cor:alon-bourgain}
Let $\mtx{A}$ be a matrix.  Then
$$
\rank(\mtx{A}) \geq \frac{\pnorm{S_1}{\mtx{A}}^2}{ \fnormsq{\mtx{A}} }
\quad\text{and}\quad
\rank(\mtx{A}) \geq \frac{\fnormsq{\mtx{A}}}{\normsq{\mtx{A}}}.
$$
\end{cor}

Alon has applied the first estimate 
in his work on extremal combinatorics~\cite{Alo03:Problems-Results}.  The second estimate
arises in a paper of Bourgain and Tzafriri on restricted invertibility~\cite{BT87:Invertibility-Large}.

\subsection{A Schur Complement Rank Identity}

Suppose that $\mtx{X}$ is a psd matrix, partitioned so that its diagonal blocks are square: 
$$
\mtx{X} = \begin{bmatrix} \mtx{A} & \mtx{B} \\ \mtx{B}^\adj & \mtx{C} \end{bmatrix}.
$$
Provided that the block $\mtx{A}$ is nonsingular, the \term{Schur complement} of $\mtx{A}$ in $\mtx{X}$ is the matrix
$$
\mtx{X}/\mtx{A} = \mtx{C} - \mtx{B} \mtx{A}^{-1}\mtx{B}^\adj.
$$
For our purposes, the relevant fact is that
\begin{equation} \label{eqn:schur-rank-ident}
\rank( \mtx{X} ) = \rank( \mtx{A} ) + \rank( \mtx{X} / \mtx{A} ).
\end{equation}
See~\cite[Sec.~2]{PSWZ07:Huas-Matrix} for more Schur complement identities.

\section{Sparsity Gap under Strong Incoherence} \label{sec:strong-incoherence}

Corollary~\ref{cor:generic-ld} indicates that we can obtain uncertainty principles for generic signals by developing lower bounds on the rank of a subdictionary $\Fee_R$. This section describes the simplest approach to this problem, which proceeds via Corollary~\ref{cor:alon-bourgain}.  This method is most effective when the coherence $\mu$ is small.

Let $R$ be a set of atoms.  Since
\begin{equation} \label{eqn:rank-gram}
\rank( \Fee_R^\adj \Fee_R ) = \rank( \Fee_R ),
\end{equation}
we may as well work with the Gram matrix of $\Fee_R$.  This substitution allows us to exploit geometric information about the dictionary.  Indeed, the diagonal entries of $\Fee_R^\adj \Fee_R$ equal one because the atoms have unit $\ell_2$ norm, and the off-diagonal entries are bounded in magnitude by $\mu$ because they contain the inner products between distinct atoms.

\begin{lemma} \label{lem:simple-result}
Let $R$ be a set of $r$ atoms.  Then
$$
\rank( \Fee_R ) \geq \frac{r}{1 + (r - 1) \mu^2}.
$$
\end{lemma}

\begin{IEEEproof}
Relation~\eqref{eqn:rank-gram} and Corollary~\ref{cor:alon-bourgain} imply that
\begin{equation} \label{eqn:rank-alon}
\rank(\Fee_R) = \rank(\Fee^\adj \Fee_R)
	\geq \frac{\pnorm{S_1}{\Fee_R^\adj\Fee_R}^2}{\fnormsq{\Fee_R^\adj \Fee_R}}.
\end{equation}
Owing to the properties of the Gram matrix,
$$
\pnorm{S_1}{ \Fee_R^\adj \Fee_R}^2 = (\trace \Fee_R^\adj \Fee_R )^2 = r^2
$$
and
$$
\fnormsq{ \Fee_R^\adj \Fee_R } \leq r + r(r-1)\mu^2
$$
Introduce these bounds into~\eqref{eqn:rank-alon} to complete the proof.
\end{IEEEproof}

When the coherence $\mu$ is small, Lemma~\ref{lem:simple-result} provides excellent estimates for the rank of $\Fee_R$.  By combining this bound with Corollary~\ref{cor:generic-ld}, we obtain our first main result.

\begin{thm} \label{thm:simple-result}
Suppose that $S$ indexes a linearly independent set of $s$ atoms and that $T$ lists $t$ atoms.  Assume that the size $\delta$ of their intersection
\begin{equation} \label{eqn:overlap-bd}
\delta = \abs{S \cap T} < s \left[1 - \frac{t-1}{s} \cdot \frac{t \mu^2}{1 - t\mu^2} \right].
\end{equation}
Let $\vct{u} = \Fee_S \vct{x}$ be a generic signal in $\range(\Fee_S)$.  Then
$$
\Prob{ \vct{u} \in \range(\Fee_T) } = 0.
$$
A fortiori, there is zero probability that $\vct{u}$ can be represented using \emph{any} set $T$ of $t$ atoms whose overlap with $S$ equals $\delta$.
\end{thm}


\begin{IEEEproof}
Define the set $R = S \cup T$.  Observe that
$$
r := \abs{R} = \abs{S} + \abs{T} - \abs{S \cap T} = s + t - \delta.
$$
Corollary~\ref{cor:generic-ld} states that
$t < \rank( \Fee_R )$
implies
\begin{equation} \label{eqn:simple-*}
\Prob{ \vct{u} \in \range(\Fee_T) } = 0. 
\end{equation}
According to Lemma~\ref{lem:simple-result},
$$
\frac{r}{1 + (r-1)\mu^2} \leq \rank( \Fee_R ).
$$
Therefore, another sufficient condition for~\eqref{eqn:simple-*} is
$$
t < \frac{r}{1 + (r-1) \mu^2} = \frac{s + t - \delta}{1 + (s+t-\delta - 1) \mu^2}.
$$
Solving this relation for $\delta$ results in the bound
$$
\delta < s - \frac{(t-1) t \mu^2}{1 - t\mu^2}.
$$
If $\delta$ satisfies this condition, then \eqref{eqn:simple-*} is in force.

Since only a finite number of sets $T$ meet the hypotheses of the theorem, there is zero chance that $\vct{u}$ is representable using \emph{any} such set $T$ of atoms.
\end{IEEEproof}

In words, Theorem~\ref{thm:simple-result} states that a generic signal constructed using a set of $s$ atoms almost surely has no representation using another set of $t$ atoms unless there is a substantial overlap between the two sets.

For example, suppose $\mu \leq m^{-1/2}$, and let $\vct{u}$ be a generic signal in $\range(\Fee_S)$, where $\abs{S} \leq m/3$.  Take $t = s$ in the overlap condition~\eqref{eqn:overlap-bd} to see that, almost surely, \emph{every} representation of $\vct{u}$ with $s$ atoms requires at least $s/2$ atoms from $S$!  

To obtain an uncertainty principle, note that the relation~\eqref{eqn:overlap-bd} is quadratic in $t$.  By reverting this inequality, we obtain a sufficient condition on $t$ as a function of $s$ and $\delta$.

\begin{cor} \label{cor:simple-result}
With the notation of Theorem~\ref{thm:simple-result}, assume
$$
t < (s-\delta -1) \left[ \sqrt{ \left(1 + \frac{1}{s-\delta-1}\right) \frac{\mu^{-2}}{s-\delta-1}  + \frac{1}{4} } - 1 \right].
$$
Then $\Prob{ \vct{u} \in \range(\Fee_T) } = 0$.
\end{cor}


After a considerable amount of algebra, Corollary~\ref{cor:simple-result} simplifies to a new uncertainty principle.

\begin{cor}[Uncertainty Principle for Generic Signals] \label{cor:tropp-up}
Suppose $S$ is linearly independent.  A generic signal $\vct{u}$ in $\range(\Fee_S)$ almost surely has no representation over a set $T$ with overlap $\delta = \abs{S \cap T}$ \emph{unless}
$$
\abs{S} + \abs{T} > \delta + \mu^{-1}\sqrt{\abs{S} - \delta}.
$$
\end{cor}

Theorem~\ref{thm:strong-intro} follows from Corollary~\ref{cor:tropp-up} by setting $\delta = 0$.

\section{Sparsity Gap under Weak Incoherence} \label{sec:weak-incoherence}

In this section, we use a different technique to bound the rank of $\Fee_R$ below.  This approach relies on the Schur complement identity~\eqref{eqn:schur-rank-ident} and the second part of Corollary~\ref{cor:alon-bourgain}.  We also require some information about the properties of random sets of atoms, drawn from~\cite{Tro08:Conditioning-Random,Tro08:Norms-Random}.  The conclusions are most interesting for weakly incoherent dictionaries.

Assume $S$ is linearly independent, which implies that the Gram matrix $\Fee_S^\adj \Fee_S$ is invertible.  Let $V$ be an index set disjoint from $S$, and write $R = S \cup V$.  Then the Gram matrix has the block structure
$$
\Fee_R^\adj \Fee_R = \begin{bmatrix}
	\Fee_S^\adj \Fee_S & \Fee_S^\adj \Fee_V \\
	\Fee_V^\adj \Fee_S & \Fee_V^\adj \Fee_V
\end{bmatrix}.
$$
The identity~\eqref{eqn:schur-rank-ident} shows that we can control the rank of the Gram matrix by controlling the rank of the northwest block and its Schur complement.  To state the result, we recall that
\begin{equation} \label{eqn:orth-proj}
\mtx{P}_S = \Fee_S (\Fee_S^\adj \Fee_S)^{-1} \Fee_S^\adj
\end{equation}
is the orthogonal projector onto $\range(\Fee_S)$.

\begin{lemma} \label{lem:schur-rank-decomp}
Suppose that $S$ is linearly independent; let $V$ be disjoint from $S$; and define $R = S \cup V$.  Then
$$
\rank(\Fee_R) = \abs{S} + \rank( (\Id - \mtx{P}_S) \Fee_V ).
$$
\end{lemma}

\begin{IEEEproof}
The Schur complement identity~\eqref{eqn:schur-rank-ident} shows that
$$
\rank(\Fee_R^\adj \Fee_R) = \rank( \Fee_S^\adj \Fee_S ) + \rank(\Fee_R^\adj \Fee_R / \Fee_S^\adj \Fee_S).
$$
The set $S$ is linearly independent, so
$$
\rank(\Fee_S^\adj \Fee_S) = \rank(\Fee_S) = \abs{S}.
$$
Next, recall the definition of the Schur complement:
$$
\Fee_R^\adj \Fee_R / \Fee_S^\adj \Fee_S
	= \Fee_V^\adj \Fee_V - \Fee_V^\adj \Fee_S (\Fee_S^\adj \Fee_S)^{-1} \Fee_S^\adj \Fee_V.
$$
Identify the orthogonal projector onto $\range(\Fee_S)$ to see that
$$
\Fee_R^\adj \Fee_R / \Fee_S^\adj \Fee_S = \Fee_V^\adj (\Id - \mtx{P}_S) \Fee_V.
$$
We conclude that
$$
\rank( \Fee_R^\adj \Fee_R / \Fee_S^\adj \Fee_S )
	= \rank( (\Id - \mtx{P}_S) \Fee_V )
$$
because $(\Id - \mtx{P}_S)^2 = (\Id - \mtx{P}_S)$.
\end{IEEEproof}

The next result provides a lower bound on the second term in Lemma~\ref{lem:schur-rank-decomp}.  Here and elsewhere, the dagger $\psinv$ denotes the pseudoinverse of a matrix.

\begin{lemma} \label{lem:schur-comp-rank}
Suppose that $S$ is linearly independent, and let $V$ be disjoint from $S$.  Then
\begin{multline*}
\rank( (\Id - \mtx{P}_S) \Fee_V ) \geq \\
	\rho^{-1} \abs{V}
	\left(1 - \smnorm{}{\Fee_S^\psinv}^2 \max\nolimits_{v \notin S} \enormsq{ \Fee_S^\adj \atom_v } \right).
\end{multline*}
\end{lemma}

\begin{IEEEproof}
Corollary~\ref{cor:alon-bourgain} states that
\begin{equation} \label{eqn:schur-bourgain}
\rank( (\Id - \mtx{P}_S)\Fee_V )
	\geq \frac{\fnormsq{(\Id - \mtx{P}_S)\Fee_V}}{\normsq{(\Id - \mtx{P}_S)\Fee_V}}.
\end{equation}
The spectral norm satisfies the bound
$$
\normsq{ (\Id - \mtx{P}_S) \Fee_V } \leq \normsq{ \Fee_V } \leq \rho
$$
because $\Fee_V$ is a column submatrix of $\Fee$.  With a little more work, we can produce a lower bound on the Frobenius norm.
\begin{align*}
\fnormsq{ (\Id - \mtx{P}_S) \Fee_V }
	&= \sum\nolimits_{v \in V} \enormsq{ (\Id - \mtx{P}_S) \atom_v } \\
	&= \sum\nolimits_{v \in V} (1 - \enormsq{ \mtx{P}_S \atom_v } ) \\
	&\geq \sum\nolimits_{v \in V} \left(1 - \smnorm{}{ \Fee_S^\psinv }^2 \enormsq{ \Fee_S^\adj \atom_v } \right) \\
	&\geq \abs{V} \left(1 - \smnorm{}{ \Fee_S^\psinv }^2 \max_{v \notin S} \enormsq{ \Fee_S^\adj \atom_v } \right).
\end{align*}
The first inequality follows by expanding the projector via~\eqref{eqn:orth-proj} and invoking the usual operator norm bound.  Introduce the two norm estimates into~\eqref{eqn:schur-bourgain} to complete the argument.
\end{IEEEproof}

Lemma~\ref{lem:schur-comp-rank} is interesting because a \emph{random} set $S$ of atoms from a \emph{weakly incoherent} dictionary usually leads to small values for the mysterious quantities in the rank bound.

\begin{prop} \label{prop:tropp}
Suppose that $S$ is a uniformly random subset of $\{1,2,\dots, s\}$ with cardinality $s$.  For $\beta \geq 1$,
$$
\Prob{ \max_{v \notin S} \enorm{ \Fee_S^\adj \atom_v }
	> \cnst{C} \beta \left[ \mu \sqrt{\log N} + 2\sqrt{\frac{\rho s}{N}} \right] } \leq 2N^{-\beta}.
$$
and
$$
\Prob{ \smnorm{}{ \Fee_S^\adj \Fee_S - \Id } < \cnst{C} \beta \left[ \mu \log N + \sqrt{\frac{\rho s \log N}{N}} \right] } \leq 2N^{-\beta}.
$$
The number $\cnst{C}$ is an absolute constant.
\end{prop}

Proposition~\ref{prop:tropp} follows from the results in~\cite[Sec.~5]{Tro08:Norms-Random} after some standard arguments.  We use Markov's inequality to convert the moment bounds into tail bounds.  Then we invoke a simple decoupling argument (see~\cite[Lem.~14]{Tro08:Linear-Independence} for a model) to transfer the result for a random set with \emph{expected} cardinality $s$ to a random set with exact cardinality $s$.  

To take advantage of Proposition~\ref{prop:tropp}, we restrict our attention to a specific class of dictionaries.  Assume that $\Fee$ is a weakly incoherent tight frame, as defined in~\eqref{eqn:weak-incoherence-intro}.  Suppose furthermore that the sparsity $s$ satisfies~\eqref{eqn:s-bound-intro}.
If we fix a sufficiently small value for constant $\cnst{c}$ and take $\beta = 1$, Proposition~\ref{prop:tropp} ensures that
$$
\max\nolimits_{v \notin S} \enorm{ \Fee_S^\adj \atom_v } \leq 1/2 \quad\text{and}\quad
\smnorm{}{ \Fee_S^\psinv } \leq \sqrt{2},
$$
except with probability $\bigO(N^{-1})$.  The upshot of this discussion is the following bound.

\begin{cor} \label{cor:weak-incoherence-rank-2}
Assume the hypotheses of Lemma~\ref{lem:schur-comp-rank} are in force.  Let $R = S \cup V$.  Then
$$
\rank( \Fee_R ) \geq \abs{S} + \frac{2m\abs{V}}{N}.
$$
\end{cor}

We reach our final major result by introducing this bound into Corollary~\ref{cor:generic-ld}.

\begin{thm} \label{thm:sparsity-gap}
Assume that $\Fee$ satisfies the weak incoherence conditions~\eqref{eqn:weak-incoherence-intro}, and assume further that $N > 2m$.  Draw a random set $S$ of $s$ atoms, where $s$ satisfies~\eqref{eqn:s-bound-intro}.  Except with probability $\bigO(N^{-1})$, the following result holds.

Suppose that $T$ lists $t$ atoms.  Let $\delta = \abs{S \cap T}$ be the size of the overlap with $S$, and assume that
$$
t < \left( s - \frac{2 \delta m}{N} \right) \left( 1 - \frac{2m}{N} \right)^{-1}.
$$
Draw a generic signal $\vct{u} = \Fee_S \vct{x}$.  Then 
$$
\Prob{ \vct{u} \in \range(\Fee_T)} = 0.
$$
A fortiori, there is zero probability that $\vct{u}$ has a representation over \emph{any} such set $T$ of atoms.
\end{thm}

\begin{IEEEproof}
Draw a random set $S$ of $s$ atoms.  Corollary~\ref{cor:weak-incoherence-rank-2} guarantees that, with probability $\bigO(N^{-1})$, the following result holds.  For any given set $T$ of $t$ atoms, define $V = T \setminus S$ and select $R = S \cup V$.  Then
\begin{equation} \label{eqn:gap-rank-bd}
\frac{2m\abs{V}}{N} \leq \rank( \Fee_R ).
\end{equation}

Now, suppose that $T$ is a specific set of $t$ atoms.  According to Corollary~\ref{cor:generic-ld}, the condition $t < \rank( \Fee_R )$ implies that a generic signal in $\range(\Fee_S)$ almost surely has no representation using the atoms in $T$.  It follows from~\eqref{eqn:gap-rank-bd} that
$$
t < s + \frac{2m\abs{V}}{N}
$$
is a stricter sufficient condition that $\vct{u}$ almost surely cannot be represented with $T$.  Introduce the identity $\abs{V} = t - \delta$ into the last relation and rearrange to obtain the equivalent condition
$$
t < \left(s - \frac{2\delta m}{N} \right) \left( 1 - \frac{2m}{N} \right)^{-1}.
$$
As usual, we may take a union bound over the finite number of admissible $T$ to obtain a uniform result.
\end{IEEEproof}


The message may be clearer if we make additional simplifications to the sufficient condition in Theorem~\ref{thm:sparsity-gap}.  First, subtract and add $2sm/N$ in the first parenthesis to reach
$$
t < \left[ s \left(1- \frac{2 m}{N}\right) + \frac{2(s - \delta)m}{N} \right]
	\left( 1 -\frac{2 m}{N} \right)^{-1}.
$$
Since $1 < (1 - 2 m/N)^{-1}$, we find that a further sufficient condition is
\begin{equation} \label{eqn:sparsity-gap}
t \leq s + \frac{2(s-\delta)m}{N}
\end{equation}
In words, there is zero probability that a generic signal in $\range(\Fee_S)$ has a representation using $t$ atoms unless $t$ is somewhat larger than $s$ or the alternative representation uses many atoms from $S$.  Take $\delta = 0$ in~\eqref{eqn:sparsity-gap} to reach Theorem~\ref{thm:weak-intro}.

\bibliographystyle{IEEEtran}
\bibliography{IEEEabrv,sparsity-gap}

\end{document}